\documentclass[english, 12pt]{article}
\author{Anirban Das, Anna Levina}
\title{Critical neuronal models with relaxed timescales separation}

\usepackage[T1]{fontenc} 
\usepackage[utf8]{inputenc} 
\usepackage{float}
\usepackage{dsfont}
\usepackage{graphicx}
\usepackage{amsmath,amsfonts,amssymb,amsthm,epsfig,epstopdf,url,array,enumerate,csquotes}
\usepackage{bm}
\usepackage{graphicx,color,psfrag}
\usepackage{color}
\usepackage{babel}
\usepackage{todonotes}
\makeatother

\usepackage{thmtools}
\theoremstyle{plain}
\newtheorem{thm}{Theorem}[section]
\newtheorem{lem}[thm]{Lemma}

\newtheorem{definition}{Definition}[section]

\newtheorem*{remark}{Remark}

\theoremstyle{definition}

\theoremstyle{remark}

\newcommand{\rom}[1]{%
  \textup{\uppercase\expandafter{\romannumeral#1}}%
}

\begin{document}
	
%
%
%
%

\title{Critical neuronal models with relaxed timescales separation}
	\author{Anirban~Das$^{1}$,  Anna~Levina$^{2,3,4}$\\
{\small $^{1}$Department of Mathematics, Pennsylvania State University, USA}\\
{\small $^{2}$University of T\"ubingen, T\"ubingen, Germany}\\
{\small $^{3}$Max Planck Institute for Biological Cybernetics, T\"ubingen, Germany } \\
{\small $^{4}$Institute of Science and Technology Austria, Klosterneuburg, Austria} 
}
\maketitle
\begin{abstract}
  Power laws in nature are considered to be signatures of complexity. The theory of self-organized criticality (SOC) was proposed to explain their origins. A longstanding principle of SOC is the \emph{separation of timescales} axiom.   It  dictates that external input is delivered to the system at a much slower rate compared to the timescale of internal dynamics. The statistics of neural avalanches in the brain was demonstrated to follow a power law, indicating closeness to critical state. Moreover, criticality was shown to be a beneficial state for various computations leading to the hypothesis, that the brain is a SOC system. However, for neuronal systems that are constantly bombarded by incoming signals, separation of timescales assumption is unnatural. Recently it was experimentally demonstrated that a proper correction of the avalanche detection algorithm to account for the increased drive during task performance leads to a change of the power-law exponent from $1.5$ to approximately $1.3$, but there is so far no theoretical explanation for this change. Here we investigate the importance of timescales separation, by partly abandoning it in various models. We achieve it by allowing for external input during the avalanche, without compromising the separation of avalanches. We develop an analytic treatment and provide numerical simulations of a simple neuronal model. If the input strength scales as one over network size we call it a moderate input regime.  In this regime, a scale-free behavior is observed i.e. the avalanche size follows a $1.25$ power law, independent on the exact size of the input. In contrast for a perfectly timescales separated system an exponent of $1.5$ is observed. Thus the universality class of the system is changed by the external input, and the change of the exponent is in a good agreement with experimental observation from non-human primates.  We confirm our analytical findings by simulations of the more realistic branching network model.

\end{abstract}

\section{ Introduction}
A  variety of natural system provide  observations that follow  power-law statistics, possibly with exponential cutoff~\cite{Gutenberg1956,sornette2006critical, Beggs2003}. For example  a power law distribution for activity propagation cascades (so-called {\it neuronal avalanches})  was reported in a myriad of neuronal systems including  cortical slices from rats~\cite{Beggs2003}, dissociated cultures~\cite{Levina2017}, \emph{in vivo} recordings in monkeys ~\cite{Petermann2009}, and humans~\cite{Priesemann2013, tagliazucchi2012criticality, shriki2013neuronal}.  In many cases, the appearance of the power-law statistic is connected with closeness to the critical point of a second order (continuous) phase transition. For the brain, the claim that power law observation point to the closeness to critical states was additionally supported by the observation of stable exponents relations~\cite{klaus2011statistical}, and shape collapse~\cite{friedman2012universal}. Models of criticality therefore began being used for  studying the brain. Additional reason for it comes from observations that criticality  brings about optimal computational capabilities ~\cite{langton1990computation, bertschinger2004}, optimal transmission and storage of information ~\cite{boedecker2012}, and sensitivity to sensory stimuli~\cite{Shew2009, Kinouchi2006}. In spite of this many facets of criticality, models often prove to be incompatible to the specific natural incident. Reconciling these two points has consumed significant effort and sparked wide debates in the neuronal community~\cite{beggs2012being, Berard_2006}.
	
 A concept of self-organized criticality (SOC) was proposed~\cite{Bak1987} as a unified mechanism for positioning and keeping systems close to criticality. SOC models have emerged as the flagship vehicle for modeling criticality as an operational
state of the brain network because it eliminates the necessity of endogenously tuning the system to criticality. For a system consisting of many interacting non-linear units, the general theory prescribes conditions necessary for exhibiting self-organized criticality.  Firstly, it should obey local energy conservation rules~\cite{Bonachela2009} and secondly, the timescale of the external drive should be separated from the timescale of interactions. It implies that no external input is delivered to the system before it reaches a stable configuration. The intuition behind the timescales separation condition can be summarized as follows: consider, there is a macroscopic scale at which the external energy is applied and a microscopic scale for  activity propagation through the interacting units. When the two scales are comparable, the frequency of the drive becomes a factor that can be tuned by some moderating party. In the limit, as the frequency of macroscopic events implodes to zero, global supervision ceases, and a self-organized system emerges~\cite{pruessner2012self, Dikman2000, dickman1998self, watkins201625}. 

  The first models of SOC in neuronal networks~\cite{Eurich2002, Herz_1995b} preceded the neuronal experiments. After experimental confirmation, further models for neuronal avalanches were developed~\cite{Arcangelis06, scarpetta2014alternation, Uhlig2013, diSanto2018landau}. Most of them included some local energy conservation, with a few rare exceptions such as the leaky integrate-and-fire neuronal model~\cite{Millman2010}.  This last model also did not have a timescales separation, there the definition of an avalanche relied heavily on the known connectivity in the network. Recently it was demonstrated~\cite{Martinello2017} that  the classical procedure of binning  will not reveal any critical statistics for this model, and neutral theory could explain the observed power-laws. The usage of binning for data-analysis from neuronal recordings~\cite{Priesemann2014,Beggs2003} implicitly relies on the assumption of timescales separation.  However, in neuronal systems inputs are constantly present and there is no chance for a strict separation of external input from the internal dynamics.  
	
We investigate here how the relaxation of the timescales separation condition by an input process influences avalanche size distribution.  An additional input to the system generally has two effects: first, the avalanches increase due to the input and follow-up firing;  second, the avalanches are  ``glued together'' namely, the input connects avalanches that would have otherwise occurred  separately. For the systems that are driven by a constant input, the definition of criticality is possible by the estimation of the branching ratio~\cite{wilting2016branching}. However, in this case typical binning-based avalanche analysis might not reveal a critical state because both aforementioned effects are present simultaneously and separating the avalanches becomes impossible.  A recent study~\cite{yu2017maintained} numerically demonstrated that changing the binning according to the firing rate reveals critical dynamics  during  task performance, when additional input on top of ongoing activity is expected. Here we investigate analytically how criticality can be preserved even if external drive is added to the system.

	As a first step towards the understanding of timescales separation, we allow for external input during avalanches without compromising their separation.  We develop an analytic treatment and provide numerical simulations of  simple neuronal models. We show that the power law scaling feature is preserved, however, even  moderate external input leads to a change in the slope of the avalanche size distribution. The same critical exponents are persistent throughout a range of values of the input.  Therefore we prove that the rate of input is not taking the role of a tuning parameter.

\section{ Models}

For our analytical and numerical investigations, we will use the following two models. The {\it Branching Model} (BM)~\cite{branching, PhysRevLett.94.058101} is a standard model to study an abstract signal propagation that serves as a simplified model for neuronal avalanches.  For our studies, we equip the standard BM with an additional
input process during avalanches.  Unfortunately BM does not allow for a complete analytic description. To overcome this difficulty, we introduce a simpler {\it Levels Model} (LM). We carry out a rigorous mathematical study of LM, and then check in simulations that similar results hold for BM. In the limit, as system size grows to infinity, both LM and BM are well approximated by branching processes~\cite{Levina2007d, LevinaDiss}.

\subsection{ The Branching Model}
The {\it Branching Model} (BM) consists of $N$ neurons connected into Erd\H{o}s-R\'{e}nyi random graph with probability of  connection  $p^{\mathrm{conn}}$. This model was used for studying benefits of criticality in neural systems~\cite{Kinouchi2006} and later was employed in many modeling investigations of neuronal avalanches~\cite{Haldeman2005,Shew2009, Levina2017,Levina2006b}.  Every edge in the network is assigned a weight $p_{ij} = \sigma/(p^{\mathrm{conn}} N)$. As a result the average sum of all outgoing weights equals $\sigma$. Each node denotes a neuron that can be in one of $n$ states, $c_i(t)$ denotes the state of the $i$-th node at the time  $t$:  $c_i =0$  indicates a resting state, $c_i =1$  is the active state, and $c_i = 2, \dots ,n-1$ are the refractory states. All states except for the active state are attained by the deterministic dynamics: if $0<c_i(t)<n-1$ then $c_i(t+1) = c_i(t) +1$, and if $c_i(t) = n-1$  then $c_i(t+1) = 0$.  

For every node $i$, the excited state $c_i(t)=1$ can  be  reached only from the resting state  $c_i(t-1)=0$ in one of the following circumstances: (1)  If a neighbor $j$ is active at time $t-1$ then with probability $p_{ji}$ $i$ will get activated at time $t$;  (2) If there are active nodes in the network, the node $i$ can receive  an external stimulus with probability ${\phi}/{N}$. 
The condition on the activity in the network in (2) is a major difference to previously studied models~\cite{Kinouchi2006}. It allows to keep the avalanche separation intact while introducing an external input during the avalanche.  We initiate the network with one random node set to the active state and the remaining nodes in the resting state and observe activity propagation ({\it avalanche}). If at a particular time step no units are in the active state, the activity propagation or avalanche is considered to have terminated and new avalanche is started. We record the distribution for activity propagation sizes  (measured in the number of activations during one avalanche), and durations (measured as the time-steps taken until activity dies out).

 It was shown~\cite{Kinouchi2006} that in the model without input, a network can exhibit different dynamical regimes depending on the value of  parameter $\sigma$ (called branching parameter): when $\sigma < 1$ the activity dies out exponentially fast, for $\sigma > 1$ there is a possibility for indefinite activity propagation. In the critical regime, obtained for  $\sigma = 1$, activity propagation size $s$ is distributed as a power-law with exponent $1.5$. However, until now it was not known, what effect additional inputs have  on these distributions.  
   
\subsection{ The Levels Model}

\begin{figure}[ht!]
	\includegraphics[width=85mm]{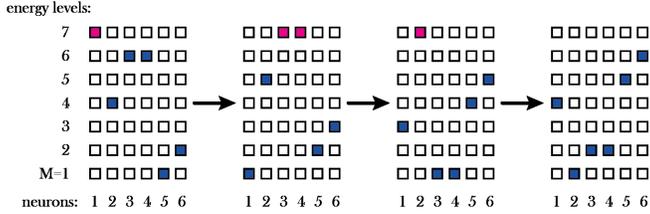}
	\caption{%
		Schematic representation of the levels model without external input, for $N=6$ neurons with $M=7$ energy levels. The avalanche size is $4$, avalanche durations is $3$.	
		\label{Fig:scheme}
	}
\end{figure}

The {\it Levels Model} (LM) without input is inspired by the simple network model of perfect integrators~\cite{Eurich_2002}. The neuronal avalanches produced by the model were shown to exhibit critical, subcritcal, and supracritical behavior depending on the control parameter, similar to experimental observations in cortical slices and cultures~\cite{Beggs2003}. The different modifications of LM were extensively studied mathematically~\cite{Levina2014, Denker2016, denker2014ergodicity, levina2008mathematical}. The version  used here was introduced in the context of dynamical systems to prove ergodicity of avalanche transformations~\cite{denker2014ergodicity} (see Appendix A). The main difference between the original biophysical model~\cite{Eurich_2002} and LM is that the later does not allow self-connections. However, when parameters are re-scaled to accommodate for changed connectivity, distributions of avalanche sizes and  durations are same in both models.

The LM consists of a fully-connected network of $N$ units, each unit $j$ is described by its energy level $E_j \in \{1,\ldots, M\}$. Connections are defined such that receiving one input  changes the energy level by $1$. In the language of neuronal modeling, $E_j$ is the membrane potential and connection strength is set to  $1$. If neuron $j$ reaches threshold level $M$, it fires a spike and then we reset it: $E_j \mapsto 1$. All neurons $k$ that  are connected to $j$ such that $E_k < M$ are updated: $E_k \mapsto E_k + 1$. After firing the spike, a neuron is set to be refractory until activity propagation is over. We initialize the model by randomly choosing energy levels of all neurons from independent copies of a uniform distribution on $[1,M]$.

After initialization, all neurons in the energy level $M$ spike, followed by dissemination of energy. If as a result more neurons reach the level $M$, then they are in turn discharged and so on,  until the activity stops. This propagation of activity we call an avalanche and the number of neurons fired is its size. The progression of the avalanche in a  system with $N=6$, and $M=7$  is demonstrated in Fig.~\ref{Fig:scheme}. 

We introduce external input to be proportional to the size of the activity propagation without input. This more sophisticated version of LM  proves to be also mathematically tractable (see Appendix B). 
If $o$ is the number of neurons fired in an avalanche, we additionally activate  $r$ among the remaining $N-o$ neurons.  Here $r$ is a random number drawn from a binomial distribution $B(o,\phi)$. The parameter $\phi \in [0,1)$ represents the rate of the external input i.e., $\phi$ is the average number of inputs delivered during an avalanche of size $1$. After these $r$ additional firings more neurons may reach the energy level $M$, resulting in a second cascade of firings. The process will stop after a maximum of $N$ discharges because no neuron is allowed to fire twice.
We study the dependence of the avalanche size distribution on the strength of the input. We use
$A_{N,\phi}$ to denote the random variable that counts the avalanche size. When $\phi=0$, we have the no-external input regime. Our model possess an Abelian property, namely it does not matter in which order to discharge neurons, size of the avalanche will be the same regardless. This property allows us to introduce external input in such a simple form.

\section{ Results and interpretation}
\subsection{ Impact of input  in LM}

\begin{figure}
	\includegraphics[width=0.5\textwidth]{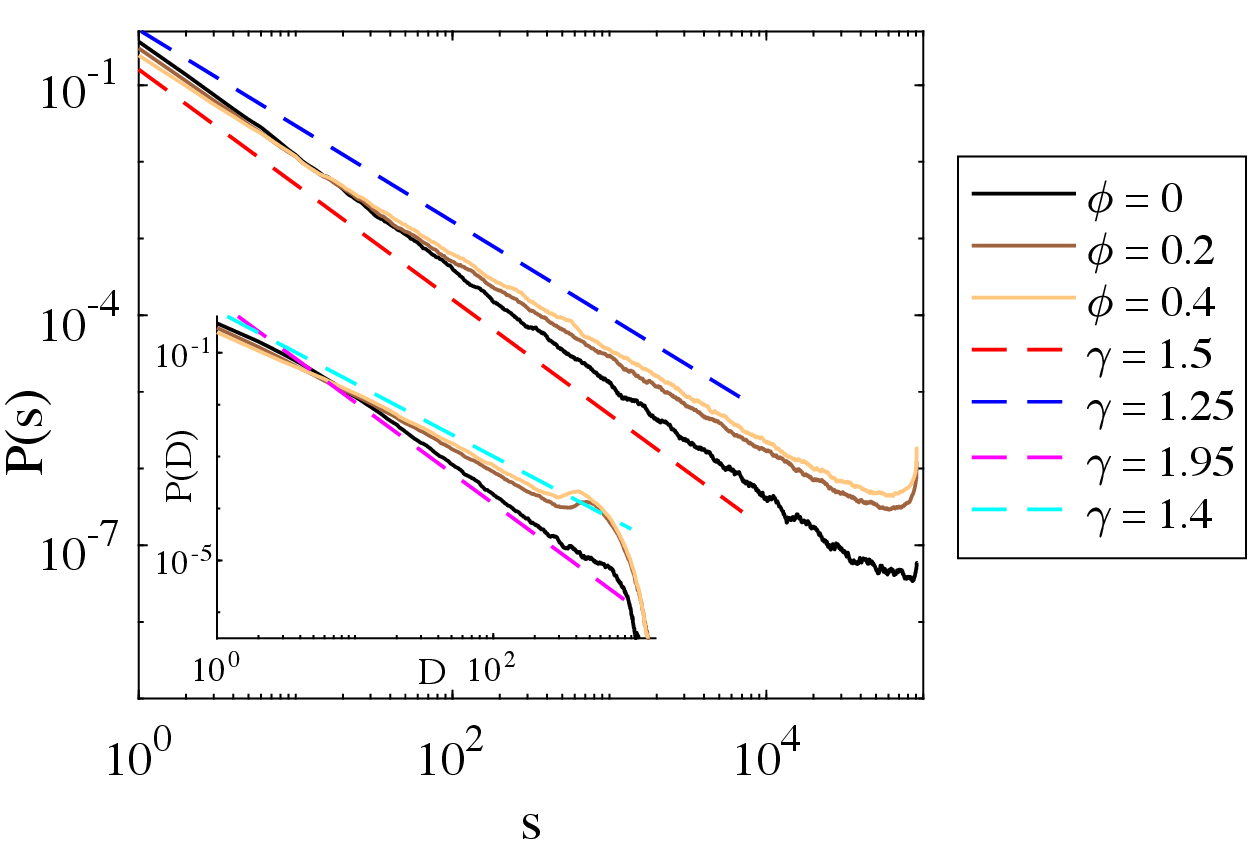}
	\caption{%
		Avalanche size distributions in LM with various input strengths. Inset shows corresponding durations distributions that also change their exponent. Input strength $\phi$ and the power-law exponents of the lines are indicated in the legend. $N = M = 10^5$.
		\label{LM_size_dur}
	}
\end{figure}

In the \enquote{no external input regime}, critical behavior is observed when $M=N$. In this case the avalanche size probability scales as a power-law i.e., $P (A_{N,0}=k) \sim C_1 k^{-1.5} $~\cite{denker2014ergodicity}.  For the rest of the article we consider $M=N$, which still serves as the critical value of the parameter in the \enquote{driven} case, with $\phi>0$.

We let $o$ denote the size of the avalanche that would have been observed without external drive, then there will be on average  $o \times \phi$ inputs. When $\phi= o(1/N)$, we can show analytically that $P (A_{N,\phi}=k) \sim C_2 k^{-1.5}$. This is the {\it small input} regime, the perturbation of the system is not strong enough to induce significant changes in the dynamics. This result demonstrates the  stability of the classical models. At the other end of the spectrum, we could  force a  fraction of the neurons to fire as a result of external input. Thus $\phi = \Theta(N)$, where $\Theta$ is taken as in the Bachmann-Landau notation~\footnote{$f(n)=\Theta(g(n))$ if $\exists k_{1}>0,\;\exists k_{2}>0,\;\exists n_{0}$ such that $\forall n>n_{0}$ we have $k_{1}\cdot g(n)\leq f(n)\leq k_{2}\cdot g(n)$}.
 In such a case we can show that $A_{N,\phi}$ converges in distribution to a  normal variable, as $N \to \infty$ (see Appendix C). Essentially,  the immense external input in this regime (named the {\it large input }regime) has reduced the neuronal activity to \enquote{noise}.

The most interesting case is the {\it moderate input } regime, where  $\phi = \Theta(1)$. In this case we can mathematically derive (see Theorem B.\ref{med_input_form2_main_asymptotic_result}) the following  result as $k$ grows to infinity: 
\begin{equation}\label{main_asymt_result}
P (A_{N,\phi}=k) \; \sim \;C_3 k^{-1.25} .
\end{equation}

We verified \eqref{main_asymt_result} by simulating a  finite LM with $N = 10^5$ neurons, and  inputs of varying   strength. As expected, the $1.5$ power-law is transformed by the input into the $1.25$ power-law (see Fig.~\ref{LM_size_dur}). Also, we numerically test  the avalanche duration distribution i.e the number of time-steps during an avalanche. Both observables deviate from the power-law in the very tail because of  finite system size and  restriction on double activation. Except for this deviations, numerical simulations support analytic results. 
 
In the moderate input regime, in spite of the compromised timescales separation, power-law scaling is preserved for both  avalanche size and duration distributions. However the power law  exponent is changed. Surprisingly, as long as $\phi = \Theta(1)$ the power-laws scaling is preserved and limiting exponent remains equal to $1.25$.   This means $\phi$  does not need to be externally tuned to achieve criticality.

\subsection{ Finite size effects and numerical simulations for LM} 
A scaling relationship given by Eq.~\ref{main_asymt_result} is valid for any  given $\phi$ if $N$ and $k$ are both large enough, and $k/N$ is small enough.  To define a more precise parameter relationship that will allow us to test results in simulations, we devise sufficient  but not necessary condition for Eq.~\ref{main_asymt_result} to hold. We require $k$ to satisfy
\begin{equation}\label{condition_on_k}
N \; \ge \; k^2. 
\end{equation}
And  we require $\phi$ to satisfy for some positive $\delta$, 
\begin{equation}\label{N_and_phi}
e^{-\left(\phi \log (N)\right)^2} \;\le \; N^{-.5-\delta} .
\end{equation}
For any $N$ and $\phi$ satisfying Eq.~\ref{N_and_phi} if $k$ is small,  we get $P( A_{N,p}=k) \sim C_4k^{-1.5}$ 
(same as for no input systems) , as $k$ grows larger we get $P( A_{N,p}=k) \sim C_3k^{-1.25}$, indicating multifractal behavior~\cite{harte2001multifractals}. 
 Approximating the stochastic input by its average,  we show that as long as $k \le \phi^{-2}$  we have $ P( A_{N,p}=k) \sim C_4k^{-1.5}$.
The simple intuition behind the multifractal behaviour is that for very small avalanches, there is substantial probability not to receive any external inputs. Thus the $1.5$ power law characteristic of traditional models with a separation of timescales is still visible.

We simulate the LM for different input strength and observe a good agreement with our analytic results, Fig.~\ref{BM_LM} solid lines. Aberrant behavior  for large avalanche sizes is due to the finite size of the system and the imposed condition that no avalanche can be larger than the system size. Theoretical prediction for the onset of the $1.25$ power-law scaling is indicated by the  magenta line, this too  is in  good agreement with numerical observations. 

\subsection{ Branching Model with input}
A BM without external input corresponds to the situation where $\phi=0$, in such a scenario the probability distribution for avalanches follow a $1.5$  power law~\cite{branching}. Here we will discuss what changes in the avalanche size distribution upon adding a moderate input.
A useful characteristic of the LM is that the avalanche can be separated into two stages, an original avalanche (pre-avalanche) and the aftershock avalanche that is triggered by external inputs. Although this feature makes the LM analytically tractable, it also makes its construction seem contrived. In contrast, in the BM external input is added at a fixed rate during the avalanches, while keeping the separation between the avalanches intact. 

In the {\it moderate input} regime, for any suitable strength of the external signal, the exponent changes from $1.5$ to $1.25$, Fig.~\ref{BM_input}. For large avalanches, finite size effects observed previously in the LM are enhanced by the possibility for the system to get additional external input during the aftershock.

\begin{figure}
	\includegraphics[width=0.5\textwidth]{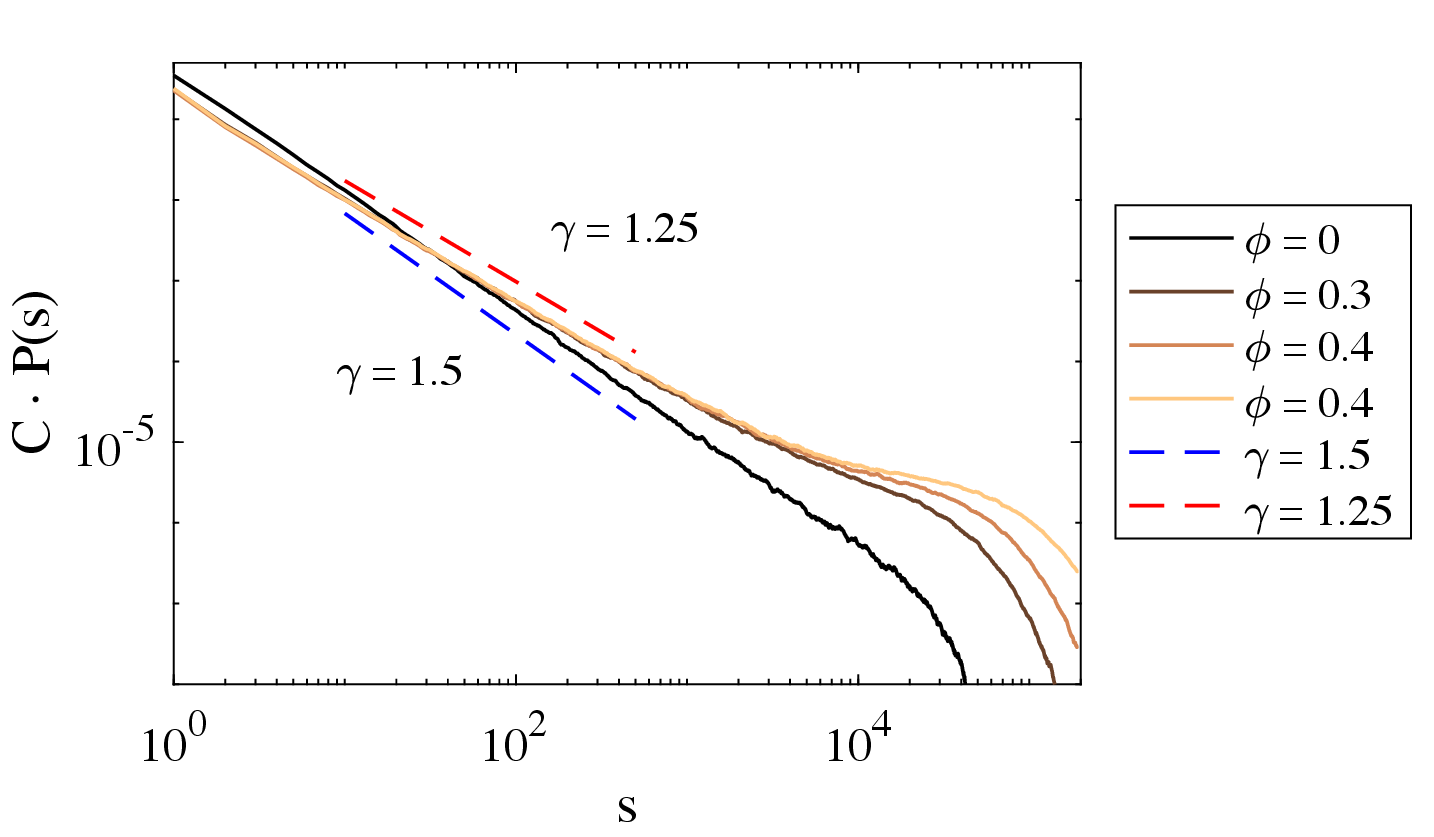}
	\caption{%
		Avalanche size distributions in the branching model with various input strengths. Input strength $\phi$ is indicated in the legend. $N = 10^5$, $n=10$, $\sigma = 1$.  Distributions for $\phi >0 $ are shifted such that they all coincide for $s = 100$.
		\label{BM_input}
	}
\end{figure}

For the BM, the input is delivered at a constant rate and is thus proportional to the duration of the pre-avalanche, while in LM the input is proportional to the size of the pre-avalanche. However, both systems show very similar avalanche size distributions for various input intensities, Fig.~\ref{BM_LM}. Let $t_r$ denote the transition time between the power-law with exponent $1.5$ and the power-law with exponent $1.25$. We observe that $t_r$ for the BM is roughly the same as for Lm, where we had analytic arguments  showed $t_r\approx \phi^{-2}$, Fig.~\ref{BM_LM}. 

\begin{figure}
	\includegraphics[width=0.5\textwidth]{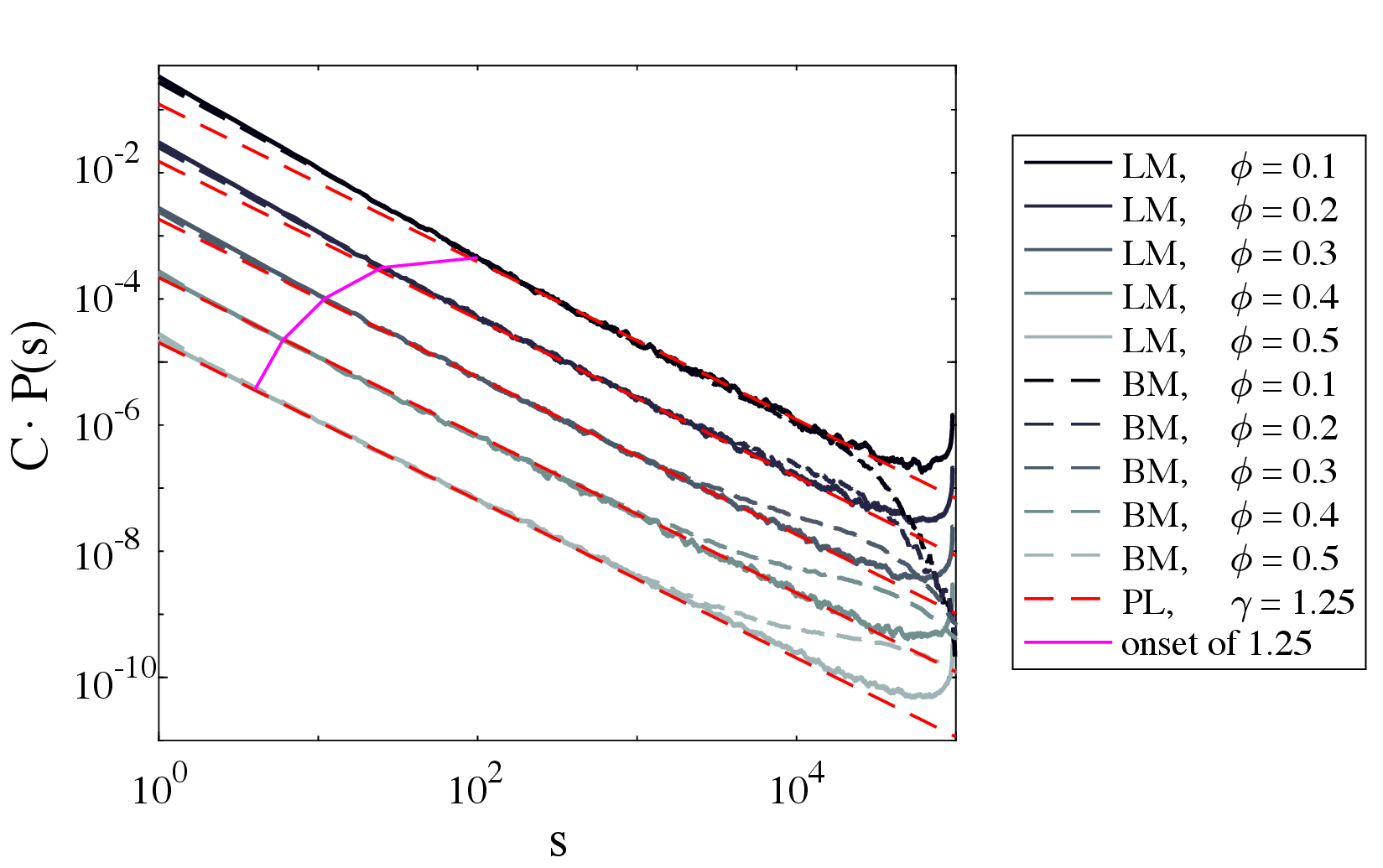}
	\caption{%
		Avalanche size distributions in the levels model and branching model with various input strengths.   Input strength $\phi$ and the model type are indicated in the legend. Magenta line indicates the analytic prediction for the onset of the power-law with exponent $-1.25$. For both models, we take $N=10^5$. To improve visibility, the distributions are shifted by multiplication with $c_{\phi} = 10^ {-10 \phi +1}$.
		\label{BM_LM}
	}
\end{figure}

\section{ Conclusion}


Models exhibiting criticality are classified into universality classes based on  power-law exponents. Quantitative characteristics of various {emergent properties} in critical systems belonging to the same universality class are found to be similar (see~\cite{christensen2005complexity}). 
 By introducing external input to models from the $1.5$ exponent universality class, we have changed them to models  with characteristic  power law exponent $1.25$. Although a $1.25$ exponent is  more seldom than the ubiquitous $1.5$ exponent, the former has been observed in several models. For example in models for {slow crack growth in heterogeneous materials}~\cite{bonamy2008crackling},  {driven elastic manifolds in disordered media}~\cite{laurson2010avalanches}, {fracturing processes under annealed disorder}~\cite{caldarelli1996self}, {mesomodels of amorphous plasticity}~\cite{PhysRevE.84.016115}, and {randomly growing networks}~\cite{mossa2002truncation}.


 Recent experimental study~\cite{yu2017maintained} has shown that task-related cortical activity is comprised of neuronal avalanches with exponent very close to our prediction. By carefully  accounting for the increase in firing rate during task performance the power law was observed to change from $1.5$ to $1.3$. The observed increase in the firing rate and nLFP rate in the premotor cortex during the task performance can be attributed to the increased input from the sensory and higher areas needed for motor planning.  There is a significant difference between models we consider and data analysis from the  recordings. Whereas in our case the ground truth about splitting the activity into avalanches is known, for the recorded data the binning procedure influences the split significantly. However, closeness of exponents obtained from interpretation of experimental recordings  to our analytic results suggests that the  adapting binning procedure is a right choice to capture underlying dynamics.

 We demonstrated that for input driven system a multi-fractal characteristics of the avalanche size distribution can be expected. Indeed, analytic approximations show that the $1.5$ power-law exponent known from the spontaneous activity analysis persists when $k \le \phi^{-2}$. This  was not seen in~\cite{yu2017maintained}, where the data collapsed to a single power law. On one hand this discrepancy could come from the difference in avalanche detection mechanisms. On the other hand it maybe that  a large input during the task  shifts the  transition to $1.25$ exponent very close to $1$ making it undetectable in the data. This hypothesis can be tested in experiments on stimulated cortical slices~\cite{Shew2009} by varying the stimulation strength and detection algorithm.

There are many open questions related to the present investigation. The most important one is: how the full elimination of timescales separation changes the outcome? In the present contribution, we did not allow for avalanches to be mixed and run parallel to each other. With simultaneous avalanches, there is no clear understanding of how one should attribute each event to any particular avalanche. Information-theoretical measures were proposed to distinguish spikes from different avalanches~\cite{williams2017unveiling}. So far the most established way to study a possible ``melange of avalanches''~\cite{Priesemann2014, wilting2016branching} is to use binning and identify empty bins to determine pauses between the avalanches. This procedure results in different power-law exponents for different bin-sizes~\cite{Beggs2003}, unless the system exhibits a true timescales separation~\cite{Levina2017}. 

Abandoning timescales separation introduces the dependence of avalanche distribution on a binning procedure. The logical hypothesis is that input during the avalanches should result in smaller power-law exponents as larger events now become more probable. Although the direction of the exponent change is easily predictable, the fact that input preserves the power-law scaling is still surprising.  Here we demonstrated this effect analytically for the levels model and numerically for branching network, but the general direction of change will remain the same for models from other universality classes. 
If we additionally allow for gluing of avalanches together it might lead to selecting smaller bin-size, than is suggested by the activity propagation timescale. This, in turn, will result in cutting of avalanches into smaller pieces and increasing the power-law exponent. This might be a reason behind the observation of exponents above $1.5$ and even around $2$ for neuronal spiking data~\cite{Friedman2012}, and LFP in \emph{ex vivo} turtle recordings~\cite{Shew2015}. Moreover, here we consider a fully connected system. However, it has been shown~\cite{yaghoubi2018neuronal} that network topology has an effect on power laws, and thus models with weaker connections can produce distributions with exponents significantly larger than $1.25$, even with additional input. 
Our result is a first step towards understanding the diversity of power-law exponents reported in the neuronal data. 

\begin{section}{Acknowledgments}
AL received funding from the People Program (Marie Curie Actions) of the 
		European Union's Seventh Framework Program (FP7/2007-2013) under REA grant 
		agreement no.~[291734].
 AL received funding from a Sofja Kovalevskaja Award
from the Alexander von Humboldt Foundation, endowed by the
Federal Ministry of Education and Research. We would also like to thank Prof. Manfred Denker for his advice and insight.
\end{section}

\begin{section}{Appendix}
Here we describe the LM purely in the language of mathematics. 
\begin{subsection}{The \enquote{$(N,p)$ BB} space}\label{section_Np_BB}
Here we introduce the $(N,p)$ BB space, this will be the central object of study for this chapter.
\begin{definition}\label{defn_Np_BB}
Given positive integers $N$ and $M$, with $M>N$, define $p=\frac{1}{M}$. The set $(N,p)$ BB is a set of $(0,1)$ matrices. A $(0,1)$ matrix $\omega$ belongs to the set $(N,p)$ BB if and only if for all $j \in \{1,2,\dots N\}$, $\sum_{i=1}^{M}a_{i,j}(\omega)=1$, where $a_{i,j}(\omega)$ denotes the $(i,j)$-th entry of $\omega$.
\end{definition}
\begin{remark}
\begin{enumerate}
\item[\textbf{i.}] The parameters $N$ and $M$ are freely chosen (but for the constraint $M >N$),  $p$ is derived from $M$. However the name $(N,p)$ BB bears the term $p$, and not $M$, this is in deference to classical considerations(\cite{levina2007critical}). Another quantity that is used in relation to the $(N,p)$ BB  space is $\alpha=\frac{N}{M}$. Whenever we speak of a $(N,p)$ BB, we assume we are speaking in terms of some $N$ and $p$ satisfying the conditions discussed here.
\item[\textbf{ii.}] The set $(N,p)$ BB is finite and we can equip it with the course sigma algebra. The set $(N,p)$ BB equipped with this sigma algebra is called the $(N,p)$ space. The elements of the set $(N,p)$ BB  are referred to as configurations. Throughout the chapter, we reserve quantities like $\omega, \omega'$ etc to denote configuration in the $(N,p)$ BB space.
\item[\textbf{iii.}] For all $i \in \{1,2,\dots M\}$ , and for all $j \in \{1,2,\dots, N\}$, $a_{ij}$ is a map between the  $(N,p)$ BB space and the set $\{0,1\}$. We will in the course of this chapter equip the $(N,p)$ space with various probability measures, in the presence of each probability measure $a_{ij}$ is a random variable. Therefore we call maps between the $(N,p)$ space and $\mathbb{R}$ (like $a_{ij}$) as universal random variables. Abusing notation we use the term random variable in place of  universal random variable, leaving the distinction to be understood by the reader.
\item[\textbf{iv.}] BB refers to \enquote{Balls and Baskets}. This because a configuration $\omega$ can be a interpreted as an array of baskets. $a_{ij}(\omega)=1$ means the basket placed at the $(i,j)$-th position of the grid contains a ball, $a_{ij}(\omega)=0$ means the basket placed at the $(i,j)$-th location of the array is empty. This intuition is not revisited in the article.
\end{enumerate}
\end{remark}

 The motivation for constructing the $(N,p)$ BB space comes from neuroscience. We  think of a configuration $\omega$ as a record the energy levels of $N$ neurons at some moment of time. Each neuron occupies one of $M$ energy levels, if $a_{ij}(\omega)=1$, then the neuron $j$ is at the $i$-th energy level. Since for all $j \in \{1,2,\dots N\}$, there is a unique $i$, such that $a_{ij}(\omega)=1$, we ensure that at any instant a neuron has one unique energy level. Formally for $j \in \{1,2,\dots N\}$, $E_{j}^{\text{el}} (\omega)= \inf_{i}\{a_{i,j}(\omega)=1\} $ ($=\sup_{i}\{a_{i,j}(\omega)=1\}$). $E_{j}^{\text{el}} (\omega)$ documents the energy level of the $j$-th neuron. There is a linear ordering of the $M$ possible energy levels, which means $E_{j}^{\text{el}} (\omega)=M$ indicates that neuron $j$ is at the highest energy level. Throughout the article as we introduce various abstract artifacts, we will try to present a parallel commentary on their interpretation from the neuronal point of view.

For all $i \in \{1,2,\dots M\}$, $Y_{i}(w)\;:=\; \sum_{j=1}^{N}a_{i,j}(\omega)$, $Y_{i}(w)$ accounts for the number of neurons at the energy level $i$.  Define the random variable $ A_{N,p}$ by $$ A_{N,p}\; := \; \inf \left\{ i|\;i \ge0,\; \sum\limits_{j=M}^{M-i} Y_j \le i\right\} .$$ $A_{N,p}$ is called the avalanche size. The motivation for considering such a random variables comes from biology,  when the neurons are at the highest energy level, they fire thus spreading all their energy uniformly to the other neurons. Each neuron on account of this internal energy being delivered climbs up to the next highest energy level. Because of one the one initial firing the other neurons may get energized to the highest level, thereby firing themselves. A series of such firings is called an Avalanche. The random variable $A_{N,p}$ gives the avalanche size (number of neurons involved in one consecutive sequence of firings). In this spirit we often say that a configuration $\omega$ has  generated $A_{N,p}(\omega)$ firings. The following sets are constructed from a configuration $\omega$,
\begin{align*}
\mathcal{F}_{\text{ire}}(\omega)\;:=\;\bigg\{j| \;a_{i,j}(\omega)=1, \; \text{for some}\; i \ge M- A_{N,p}(\omega)\bigg\}, \\ \mathcal{NF}_{\text{ire}}(\omega)\;:=\;\bigg\{j| \;a_{i,j}(\omega)=0, \; \forall\; i \ge M- A_{N,p}(\omega)\bigg\}.
\end{align*}

 We will equip the $(N,p)$ BB space with various probability measures. Each such probability measure arises from biological motivations. The first and simplest is what we call the uniform measure, we denote it by $P$.  
  It is defined as follows: $\Hat{\kappa}^{M}$ be a random variable taking values in $\{1, 2,\dots, M \}$, which is uniformly distributed. $\Hat{\kappa}^{M}_j,\; j \in \mathbb{Z}^+$, be iid copies of $\Hat{\kappa}^{M}$ defined on some probability space $\Hat{\mathcal{\kappa}}$.  The  map $C_{\text{Uf}}: \Hat{\mathcal{\kappa}} \mapsto (N,p) \; \text{BB}$ is defined as  $a_{i,j}(C_{\text{Uf}} (\theta))= \mathbf{1}_{i}\left(\kappa^{M}_j (\theta)\right)$, here $\mathbf{1}$ is the indicator function. $P$ is the push forward measure of $C_{\text{Uf}}$ on $(N,p)$ BB. Intuitively with the uniform measure every neuron has equal probability of lying in any of  the energy levels, also there is no correlation between the energy levels of different neurons.
  In \cite{denker2014ergodicity} one finds 
\begin{equation}{\label{av_distribution}}
P (A_{N,p}=k)= {{N}\choose{k}}\;p^k\;(1-(k+1)p)^{N-k} (k+1)^{k-1} .
\end{equation}
  For the remainder of this article a random variable following such a distribution will be said to have the Avalanche distribution. Before ending the section, we prove the following result that enumerates the number of configurations satisfying a given property.
\begin{thm}{\label{Basic_enumeration_theorem_trees}}
Define the set of configurations $\Game(N,k,a,p)$ as $$\Game(N,k,a,p)\; :=\; \{\omega \; | \;\omega \in (N,p)\;\text{BB},\; A_{N,p}(\omega)=k, \; Y_M(\omega)=a \} .$$
Then $|\Game(N,k,a,p)|={{k-1} \choose {a-1}} k^{k-a}$.
\end{thm}
\begin{proof}
Define the sets $V:= \{R,1,2,3, \cdots k\}$, $\Game(N,k,p):= \{\omega\;|\; \omega \in (N,p)\;\text{BB},\; A_{N,p}(\omega)=k\}$. $\mathcal{T}^{l}_{V}$  denote the set of labeled trees which have $V$ as it's set of vertices.  We will define a function  $\Psi: \Game(N,k,p) \rightarrow \mathcal{T}^{l}_{V} $. For $\omega \in \Game(N,k,p)$, here is how define $\Psi(\omega)$:

 $|\mathcal{F}_{\text{ire}}(\omega)|=k$, we first introduce a ranking for the members of $\mathcal{F}_{\text{ire}}(\omega)$. Formally rank($\cdot, \omega$) is a one-one function between $\mathcal{F}_{\text{ire}}(\omega)$ and $\{1,2,\dots k\}$.  Say $i \in \mathcal{F}_{\text{ire}}(\omega)$ and $E_{i}^{\text{el}} (\omega)=i'$, define score$(i) = k \times (M-i') + i$. The elements of $\mathcal{F}_{\text{ire}}(\omega)$ are ordered (ranked) linearly according to the inverse of their scores. This means for any $i* \in \mathcal{F}_{\text{ire}}(\omega)$, rank$(i^*, \omega) \;=\;|\{i| i \in \mathcal{F}_{\text{ire}}(\omega),\; \text{score}(i) \le \text{score}(i^*) \}|$.  
 Now $\forall u \in V$, such that  such that $E_{u}^{\text{el}} (\omega)=M$  attach (draw an edge between) $u$ and $R$ in  $\Psi(\omega)$. Further,
  $\forall i,j$ such that for some $r$ , rank$(i, \omega)=r $, $E_{j}^{\text{el}} (\omega)=M-r$, attach $i$ to $j$ in $\Psi(\omega)$.

It is straightforward to prove that the graph  $\Psi(\omega)$ is a tree, and that  the map $\Psi$
is both injective and surjective. We know from Cayley's theorem (\cite{moon1967various}) that the number of labelled trees with $k+1$ vertices is $(k+1)^{k-1}$. So we have established that in a $(N,p)$-BB space the number of configurations $\omega$ such that $A_{N,p}(\omega)=k$ is $(k+1)^{k-1}$. This can be used to prove \eqref{av_distribution}.

Note that a configuration $\omega$ has the property that $A_{N,p}(\omega)=k, \; Y_M(\omega)=a$ if and only if $R$ is joined to exactly $a$ neighbors in $\Psi(\omega)$. The number of such configurations has been computed to be ${{k-1} \choose {a-1}} k^{k-a}$ (see \cite{moon1967various}).
\end{proof}
\end{subsection}

\begin{subsection}{A technical model}
This section introduces a simple probability measure on the ($N,p$) BB space, this new probability space will help with computations that arise in future sections. The results of this section therefore are not interesting by themselves, but will serve as  tools in later efforts.

Suppose we start with a ($N,p$) BB space. There are $N$ neurons, each lying in one of $M$ ($M > N$) energy levels. We consider $p=\frac{1}{M}= \frac{\alpha}{N}$, $\alpha \le 1$. Previously we had an uniform measure $P$ on this space, i.e. each neuron was placed independently with equal chance of being in one of the $M$ energy levels. For $\lambda$ an integer valued parameter, we will construct a second probability measure on the ($N,p$) BB space. This measure denoted by $P_{\lambda}$ is the pushforward measure of $P$ by the map $T_{\lambda}: (N,p)\;\text{BB} \to (N,p)\; \text{BB}$, defined below. 

Take a configuration $\omega$. Here is the configuration $T_{\lambda}(\omega)$: \newline
\textbf{i.} For  $j$ such that $E_{j}^{\text{el}} (\omega)  \ge M-\lambda$,   
$$a_{M,j}\left(T_{\lambda}(\omega)\right)\;=1\quad \& \quad  a_{i,j}\left(T_{\lambda}(\omega)\right)\;=0, \; \forall \; i <M.
$$ 
\textbf{ii.} For $j$ such that $E_{j}^{\text{el}} (\omega)  < M-\lambda$, define $Sh_{j,\lambda}(\omega)= E_{j}^{\text{el}} (\omega)+ \lambda$ (we shall suppress the subscripts in $Sh_{j,\lambda}$ for convenience), and set
$$a_{Sh (\omega), \; j}\left(T_{\lambda}(\omega)\right)\;=1\quad \& \quad  a_{ij}\left(T_{\lambda}(\omega)\right)\;=0, \; \forall \; i \ne Sh(\omega).
$$ 
\begin{thm}{\label{Basic_enumeration_theorem_moderate_input}}
Let $A_{N,p}$ be the avalanche random variable on the ($N,p$) BB space, and $k, \lambda$  be non negative integers satisfying  $k+\lambda+1 < M$. We have
\begin{eqnarray*}
&P_{\lambda}(A_{N,p}=k)= (\lambda+1){{N}\choose{k}}\;p^k\;(1-(k+1+\lambda)p)^{N-k} (k+1+\lambda)^{k-1}  \\
&P\left( A_{N,p}(\;T_{\lambda}(\omega)\;)=k , Y_M(\omega)=0\; \right)= (\lambda){{N}\choose{k}}\;p^k\;(1-(k+1+\lambda)p)^{N-k} (k+\lambda)^{k-1}. \; 
\end{eqnarray*}
\end{thm}
\begin{proof}
$\Game(N,k,a,p, \lambda)= \{\omega| A_{N,p}(\omega)=k, \; Y_M(\omega)=a \} \cap \text{Range}(T_{\lambda})$. When $\omega \in \Game(N,k,a,p, \lambda) $, there are exactly $(\lambda+1)^{a}$ configurations $\omega'$, such that $T_{\lambda}(\omega')=\omega$. Using Theorem \eqref{Basic_enumeration_theorem_trees}, we get
\begin{equation*}
P_{\lambda} (\Game(N,k,a,p, \lambda)) = {{N}\choose{k}}p^k(1-(k+\lambda+1)p)^{N-k}{{k-1} \choose {a-1}} k^{k-a} (\lambda+1)^{a}.
\end{equation*}
\begin{eqnarray*}
P_{\lambda}(A_{N,p}=k)&= \sum_{a=1}^{k}{{N}\choose{k}}p^k(1-(k+\lambda+1)p)^{N-k}{{k-1} \choose {a-1}} k^{k-a} (\lambda+1)^{a}\\
&=(\lambda+1){{N}\choose{k}}\;p^k\;(1-(k+1+\lambda)p)^{N-k} (k+1+\lambda)^{k-1}.
\end{eqnarray*}

\end{proof}
\end{subsection}

\begin{subsection}{A model with moderate external input}\label{A model with moderate external input}
We will construct yet another measure on the $(N,p)$ BB space, we shall call it  $P^{E,\text{med}}_{\phi}$. For any real number $\phi$ satisfying $0<\phi \le 1$, we will define the random function $\tau_{\phi}:  (N,p) \; \text{BB} \to (N,p)\; \text{BB}$. For any $\omega \in (N,p)$ BB, $\tau_{\phi}(\omega)$ is defined as :  \newline
\textbf{i.} Say $A_{N,p}(\omega)= o$, when $N-o \ge o \phi$, we choose a subset of size $\lceil o\times \phi  \rceil$  from $\mathcal{NF}_{\text{ire}}(\omega)$. The chosen set be denoted by $\mathcal{EF}_{\text{ire}}(\omega)$. If $N-o < \lceil o\times \phi  \rceil$, define $\mathcal{EF}_{\text{ire}}(\omega)=\mathcal{NF}_{\text{ire}}(\omega)$.\newline
\textbf{ii.} When $j \in \mathcal{EF}_{\text{ire}}(\omega)$
\begin{align*}
a_{M,j}\left(\tau_{\phi}(\omega)\right)\;=1\quad \& \quad  a_{ij}\left(\tau_{\phi}\right)\;=0, \; \forall \; i <M.
\end{align*}
When $j \notin \mathcal{EF}_{\text{ire}}(\omega) $, set $a_{i,j}\left(\tau_{\phi}(\omega)\right)\;=\; a_{i,j}(\omega)$, for all $i$.\newline
$P^{E,\text{med}}_{\phi}$ is the pushforward measure of $P$ by
$\tau_{\phi}$.
The intuition behind the definition of $\tau_{\phi}$, is as follows : during the avalanche we want to introduce some external signals to the system. The number of these external signals is $|\mathcal{EF}_{\text{ire}}(\omega)|$, and the neurons receiving external input are those whose numbers lie in the set $\mathcal{EF}_{\text{ire}}(\omega)$. The intricacy here is that the size of the set $\mathcal{EF}_{\text{ire}}(\omega)$ is $\lceil o\times \phi  \rceil$, this means that the external input is proportional to the size of the original avalanche. The longer the  avalanche the more the number of external stimulus's delivered during it.
\begin{thm}{\label{med_input_form1}}
Let $A_{N,p}$ be the avalanche random variable on the $(N,p)$ BB space, $\phi$ and $\tau_{\phi}$ are as above. For any positive integers  $k, o$ satisfying  $k \ge o+\lceil o\times \phi  \rceil$, and for $\hat{p}= \frac{p\times N}{N-o}$ we have
\begin{eqnarray*}
&P( A_{N,p}(\tau_{\phi}(\omega))=k\;|\;  A_{N,p} (\omega)=o \;)= \\ &(\lceil o\times \phi  \rceil){{N-\lceil o\times \phi  \rceil-o}\choose{k-\lceil o\times \phi  \rceil-o}}\;\hat{p}^{k-\lceil o\times \phi  \rceil-o}\;(1-(k-o+1)\hat{p})^{N-k} (k-o)^{k-1-\lceil o\times \phi  \rceil-o} .
\end{eqnarray*}

\end{thm}

\begin{remark}
 We will study the regime as $N,k \to \infty, N >> k$. We consider $\phi >0$ to be a constant, i.e we consider $\phi N \to \infty$. Hence many of the formulas will fail to hold when one directly sets $\phi=0$, and compares with results in the no input regime where we use the measure $P$. 
\end{remark}

\textbf{\underline{Asymptotics}}\newline

The main result here shows that $P^{E,\text{med}}_{\phi}( A_{N,p}=k)$ becomes a power law as $N$ tends to $\infty$.  Since we are dealing with asymptotic behavior, we will for clarity replace  $\lceil o\times \phi  \rceil$ with $o\times \phi$. The introduction of this simplification has no bearings on the final result.
We will first establish some lemma's.
\begin{lem}\label{technical_lemma_moderate_input_2}
The following is true for $N, k$ positive integers satisfying $N>k$, and $p=\frac{1}{N+1}$,
\begin{eqnarray*}
\lim_{\substack{k \to \infty\\ N \ge k^2}} \frac{ {{N}\choose{k}}\;p^k\;(1-(k+1)p)^{N-k} (k+1)^{k-1} }{k^{-1.5}}=  \frac{1}{\sqrt{2 \pi}},\\
\lim_{\substack{k \to \infty\\ N \ge k^2}} \frac{ {{N}\choose{k}}\;p^k\;(1-(k+1)p)^{N-k} (k)^{k-1} }{k^{-1.5}}=  \frac{1}{\sqrt{2 \pi}\times e}.
\end{eqnarray*}
\end{lem}
\begin{lem}\label{technical_1_1}
Let $k,o,o',g$ be positive integers satisfying, $k=o+o'+g$, and  $0 \le \frac{o'}{o} \le 1$ . Then 
\begin{eqnarray*}
\lim_{\substack{k \to \infty\\  \frac{g}{o'} \to \infty}}\frac{(k-o)^{o'}}{(k-o+1)^{o'}}=1\bigg(1 +O(\frac{1}{g})\bigg),\\
\lim_{\substack{k \to \infty\\  \frac{g}{o'} \to \lambda}}\frac{(k-o)^{o'}}{(k-o+1)^{o'}}=e^{-\frac{1}{1+\lambda}}\bigg(1 +O(\frac{1}{g})\bigg).
\end{eqnarray*}
\end{lem}

\begin{remark}
Typically we will apply \ref{technical_1_1} with $o'=o\phi$. The two parts tell us how to deal with the asymptotics in the respective cases where the original avalanche is very small compared to the whole avalanche, and when it is not.    Continuing upon the remark following Theorem \ref{med_input_form1}, observe that the setting $g=o'=0$ is beyond the scope of Lemma \ref{technical_1_1}, this is a setting that becomes significant if we were to consider $\phi=0$.
\end{remark}

\begin{lem}\label{technical_lemma_moderate_input_1}
Let $N,k,o,o'$ be positive integers satisfying,  $\;k-o-o' \ge 0\;$, $\;o > 0\;$, and $\;\phi o=o'\;$, with $\;0<\phi <1\;$. Also set $\hat{p}:=\frac{1}{N-o}$. Then,
\begin{equation}{\label{technical_2}}
\lim_{\substack{o \to \infty \\ \frac{o^2}{N} \to 0 }}\frac{{{N-o'-o}\choose{k-o'-o}}\;\hat{p}^{-o'} (k-o-o')!}{{{N-o}\choose{k-o}} (k-o)!}=1.
\end{equation}
Also,
\begin{equation}{\label{technical3}}
 \frac{(k-o)!}{(k-o-o')! (k-o+1)^{o'}} \le e^{-\frac{o'(o'+1)}{2(k-o+1)}}.
\end{equation}
In addition to the conditions at the beginning of the lemma , if we assume $N= C\; {o'}^{2}$, for some constant $C >0$, we get,
\begin{equation}{\label{technical_2_star}}
1 \quad \le \quad\lim_{o \to \infty }\frac{{{N-o'-o}\choose{k-o'-o}}\;\hat{p}^{-o'} (k-o-o')!}{{{N-o}\choose{k-o}} (k-o)!} \quad \le \quad e^{C}.
\end{equation}

Further if in addition to the conditions at the beginning of the lemma , we assume that $\sqrt{k} \log(k) \ge  o$, then the following is true (asymptotics are taken in the sense $o \to \infty$ ):
\begin{equation}{\label{technical4}}
 \frac{(k-o)!}{(k-o-o')! (k-o+1)^{o'}} \sim e^{-\frac{o'^2}{2(k-o)}}e^{-R_{\text{em}}}, \; \text{where} \; |R_{\text{em}}|= \mathcal{O}(\frac{\log(k)}{\sqrt{k}}).
\end{equation}
\end{lem}
\begin{proof}
Let $(x)_n= (x)(x-1)\dots (x-n+1) $, denote the falling factorial function. We can derive the following:
\begin{equation}\label{eq_1_technical_lemma_moderate_input_1}
\frac{{{N-o'-o}\choose{k-o'-o}}\;\hat{p}^{-o'} (k-o-o')!}{{{N-o}\choose{k-o}} (k-o)!}= \frac{(N-o)^{o'}}{(N-o)_{o'}} \ge 1.
\end{equation}
Now, notice that 
\begin{equation}\label{eq_2_technical_lemma_moderate_input_2}
\frac{(N-o)^{o'}}{(N-o)_{o'}} \le \frac{(N-o)^{o'}}{(N-o-o')^{o'}} \sim e^{{o'}^2(N-o)^{-1}}.
\end{equation}
The final bound in \eqref{eq_2_technical_lemma_moderate_input_2} follows from 
\begin{align*}
\begin{aligned}
\log \left[\frac{(N-o)^{o'}}{(N-o-o')^{o'}} \right] = -o' \log[1-\frac{o'}{N-o}]= \frac{{o'}^2}{N-o}+  I_{\text{nt}},
\end{aligned}
\end{align*}
where $ \quad |I_{\text{nt}}| = \mathcal{O}\left(\frac{o^3}{(N-o)^2}\right)$.
From \eqref{eq_1_technical_lemma_moderate_input_1} and \eqref{eq_2_technical_lemma_moderate_input_2} we can derive both \eqref{technical_2} and \eqref{technical_2_star}.

Observe that $ \frac{(k-o)!}{(k-o-o')! (k-o+1)^{o'}}= \frac{(k-o)_{o'}}{(k-o+1)^{o'} }$. The proof of \eqref{technical3} is an immediate consequence of 
\begin{align*}
\log\bigg[\frac{(k-o)_{o'}}{(k-o+1)^{o'} }\bigg]= \sum_{i=1}^{o'} \log\left(1-\frac{i}{k-o+1}\right)\quad \le \\ -\sum_{i=1}^{o'}\frac{i}{k-o+1}\quad = \quad-\frac{o'(o'+1)}{2(k-o+1)}.
\end{align*}
To prove \eqref{technical4}, observe that when $\sqrt{k} \log(k) \ge  o$, there exists $0\;\le \theta \le \;1$, such that
\begin{align*}
 \sum_{i=1}^{o'} \log\left(1-\frac{i}{k-o+1}\right) \quad= \quad \sum_{i=1}^{o'}\left(-\frac{i}{k-o+1}-\frac{1}{2}\left(\frac{i}{k-o+1}\right)^2\frac{1}{1-\theta\frac{i}{k-o+1}} \right)\\
\quad = \quad \sum_{i=1}^{o'}\left(-\frac{i}{k-o+1}\right) \pm \mathcal{O}\left(\frac{\log(k)}{\sqrt{k}}\right)\quad = \quad-\frac{o'(o'+1)}{2(k-o+1)} \pm \mathcal{O}\left(\frac{\log(k)}{\sqrt{k}}\right).
\end{align*}

\end{proof}

We say that $f(k) \lnsim C $ if $\lim_{k \to \infty} f(k) \le C$.

\begin{thm}{\label{med_input_form2_main_asymptotic_result}}
 We  assume that $(N, \phi, k)$ satisfies :
\begin{enumerate}
\item[\textbf{i.}]$\phi > \frac{1}{{\log(k)}^{.5}}$ .
\item[\textbf{ii.}]$N \ge k^2$.
\end{enumerate}
Then with $P^{E,\text{med}}_{\phi}( A_{N,p}=k)$ as in section \ref{A model with moderate external input} and $p=\frac{1}{N+1}$, there exist positive constants $D_1, D_2$ depending on $\phi$ , such that
\begin{equation}
D_2 \lnsim \frac{P^{E,\text{med}}_{\phi}( A_{N,p}=k)}{k ^{-1.25}} \lnsim D_1.
\end{equation}
\end{thm}
\begin{proof}
Here is the proof of Theorem \ref{med_input_form2_main_asymptotic_result}. 
  In  light of \eqref{technical_2}, one finds :
\begin{align*}
P(\;  A_{N,p}(\tau_{\phi}(\omega))=k |\;  A_{N,p} (\omega)=o )\quad =  \quad \\
  {{N-o}\choose{k-o}}\;\hat{p}^{k-o}\;(1-(k-o+1)\hat{p})^{N-k} (k-o)^{k-1-o}\frac{(\phi o)(k-o)!}{(k-o-o\phi)! (k-o)^{o\phi}} .
\end{align*}
Now it can be shown that $P( A_{N,p} =o) \sim \frac{o^{-1.5}}{\sqrt{2\pi}}$. 
Thus,
 $$P\bigg(\;  A_{N,p}(\tau_{\phi}(\omega))=k ,\;  A_{N,p} (\omega)=o \;\bigg)= P(\;  A_{N,p}(\tau_{\phi}(\omega))=k |\;  A_{N,p} (\omega)=o \;) \frac{o^{-1.5}}{\sqrt{2\pi}}.$$
Using the Euler–Maclaurin formula for expressing sums as integrals, we get
\begin{align}\label{zeroth_part_big_result}
P^{E,\text{med}}_{\phi}( A_{N,p}=k)= & \int_{1}^{\frac{k}{1+\phi}} P\bigg(\;  A_{N,p}(\tau_{\phi}(\omega))=k ,\;  A_{N,p} (\omega)=o \bigg)do \nonumber  \\
&  +\phi \frac{1}{\sqrt{2 \pi }e^2}  k^{-1.5} + (1+\phi)^{1.5}k^{-1.5}\frac{e^{-\frac{k\phi}{1+\phi}}}{\sqrt{2 \pi }e}+o(k^{-1.5})O(\phi).
\end{align}

Now using \eqref{technical4}
\begin{align}\label{first_part_big_result}
\int_{1}^{\frac{\sqrt{k}}{\log k}} P\bigg(\;  A_{N,p}(\tau_{\phi}(\omega))=k ,\;  A_{N,p} (\omega)=o \;\bigg)do  &\sim\frac{\phi}{e 2\pi}\int_{1}^{\frac{\sqrt{k}}{\log k}} (k-o)^{-1.5} o^{-.5} do  \nonumber \\
& \sim \frac{\phi}{e\pi} \bigg(\frac{k^{-1.25}}{(\log k)^{.5}}-k^{-1.5} \bigg).
\end{align}

Observe that using the fact $\phi > (\log k)^{-.5}$, for  some positive constant $C^{(1)}$, we have,
\begin{align*}
\int_{\sqrt{k} \log k}^{\frac{k}{1+\phi}} P\bigg(\;  A_{N,p}(\tau_{\phi}(\omega))=k ,\;  A_{N,p} (\omega)=o \;\bigg)do \\ \le C^{(1)}\phi\int_{\sqrt{k} \log k}^{\frac{k}{1+\phi}}(k-o)^{-1.5} o^{-.5} e^{\frac{-(\phi o)^2}{2(k-o)}}do \\ 
\le C^{(1)} \phi   e^{-(\phi \log k)^2}\int_{\sqrt{k} \log k}^{\frac{k}{1+\phi}} o^{-.5}(k-o)^{-1.5} do  \le \phi \; C^{(1)}  e^{-(\phi \log k)^2} k^{-.75} \sim o(k^{-1.25}).
\end{align*}
Thus 
\begin{equation}\label{second_part_big_result}
\int_{\sqrt{k} \log k}^{\frac{k}{1+\phi}} P\bigg(\;  A_{N,p}(\tau_{\phi}(\omega))=k ,\;  A_{N,p} (\omega)=o \;\bigg)do= o(k^{-1.25}) .
\end{equation}

Now, Define $A_k= \frac{\phi^2}{2(k-\frac{\sqrt{k}}{\log k})},\; B_k=\frac{\phi^2}{2(k-\sqrt{k}\log k)} $, using \eqref{technical4} we get
\begin{align}\label{two_sided_bounds_for_probability_density}
\frac{\phi}{2\pi e}(k-\sqrt{k} \log k)^{-1.5} Q_b \quad \ge \quad &\int_{\frac{\sqrt{k}}{\log k}}^{\sqrt{k} \log k} P\bigg(\;  A_{N,p}(\tau_{\phi}(\omega))=k ,\;  A_{N,p} (\omega)=o \;\bigg)do \quad \ge \\ &\frac{\phi}{2\pi e}(k-\frac{\sqrt{k}}{\log k})^{-1.5} Q_s,
\end{align}
where
\begin{align*}
&Q_s= \int_{\frac{\sqrt{k}}{\log k}}^{\sqrt{k} \log k} o^{-.5} e^{-o^2 B_k} do \sim  \frac{1}{2 B_k^{.25}}\int_{\frac{\phi^2}{2}}^{\infty} t^{-.75}e^{-t}dt,\\
&Q_b= \int_{\frac{\sqrt{k}}{\log k}}^{\sqrt{k} \log k} o^{-.5} e^{-o^2 A_k} do \sim  \frac{1}{2 A_k^{.25}}\int_{\frac{\phi^2}{2}}^{\infty} t^{-.75}e^{-t}dt.
\end{align*}

Thus we arrive at 
\begin{equation}\label{third_part_big_result}
\int_{\frac{\sqrt{k}}{\log k}}^{\sqrt{k} \log k} P\bigg(\;  A_{N,p}(\tau_{\phi}(\omega))=k ,\;  A_{N,p} (\omega)=o \;\bigg)do \sim  D_\phi k^{-1.25}, \; 
\end{equation}
with $D_\phi = \frac{\sqrt{\phi}}{2^{1.75}\pi e}\int_{\frac{\phi^2}{2}}^{\infty} t^{-.75}e^{-t}dt$.
Using \eqref{zeroth_part_big_result}, \eqref{first_part_big_result}, \eqref{second_part_big_result} and \eqref{third_part_big_result}, one arrives at 
$$
P^{E,\text{med}}_{\phi}( A_{N,p}=k)= D_\phi k^{-1.25} +\frac{\phi}{e\pi} \bigg(\frac{k^{-1.25}}{(\log k)^{.5}}-k^{-1.5} \bigg) + \phi \frac{1}{\sqrt{2 \pi }e^2}  k^{-1.5} + o(k^{-1.5}).
$$

\end{proof}
\begin{remark}
The conditions enforced on $(N, \phi, k)$ in Theorem \eqref{med_input_form2_main_asymptotic_result} are sufficient but not necessary. For example the condition $N> k^2$ is used to ensure terms like $e^{\frac{o\phi}{N-o-1}}$ are equal to $1$. The milder condition $N> o^2$, would suffice for this.
\end{remark}

\textbf{\underline{Cutoff at $k> \phi^{-2}$}}\newline
Consider $\phi$ is small but fixed.  We take $k= \phi^{-\delta}$, i.e $k^{-\frac{1}{\delta}}= \phi$. A rough simplification of the proof for Theorem \eqref{med_input_form2_main_asymptotic_result} shows that the term $P^{E,\text{med}}_{\phi}( A_{N,p}=k)$ is computed as,
\begin{align}\label{between_Ranges_orginal}
\begin{aligned}
P^{E,\text{med}}_{\phi}( A_{N,p}=k) \quad \sim \quad \phi\int_{1}^{\frac{k}{1+\phi}}(k-o)^{-1.5} o^{-.5} e^{\frac{-(\phi o)^2}{2(k-o)}}do  \\
\sim \phi\bigg[{\displaystyle\int_{1}^{\sqrt{k}}}(k-o)^{-1.5} o^{-.5} do+{\displaystyle\int_{\sqrt{k}}^{\frac{k}{1+\phi}}}(k-o)^{-1.5} o^{-.5} e^{\frac{-(\phi o)^2}{2(k-o)}}do\bigg] \\
 \sim\; C \phi k^{-1.25}.
 \end{aligned}
\end{align}
The reason for breaking up the main integral in \eqref{between_Ranges_orginal} is that when $k> e^{\frac{1}{\phi}}$, and $o> \sqrt{k}$, one observes
\begin{align}\label{between_Ranges_key_conditions}
 e^{\frac{-(\phi o)^2}{2(k-o)}}\quad\sim \quad 0.
\end{align}
With $k= \phi^{-\delta}$, \eqref{between_Ranges_key_conditions} is no longer true. To account for this  we define $\delta'=\frac{1}{2}+ \frac{1}{\delta} + \delta'''$, with  any $\delta'''>0$ satisfying $ \frac{1}{\delta} \gg \delta'''$. We also need to ensure $\delta > 2$. Such a setup  allows for \eqref{between_Ranges_orginal} to be replaced with
\begin{align}\label{between_Ranges_1}
\begin{aligned}
P^{E,\text{med}}_{\phi}( A_{N,p}=k) \quad \sim \quad \phi\int_{1}^{\frac{k}{1+\phi}}(k-o)^{-1.5} o^{-.5} e^{\frac{-(\phi o)^2}{2(k-o)}}do  \\
\sim \phi\bigg[{\displaystyle\int_{1}^{k^{\delta'}}}(k-o)^{-1.5} o^{-.5} do+{\displaystyle\int_{k^{\delta'}}^{\frac{k}{1+\phi}}}(k-o)^{-1.5} o^{-.5} e^{\frac{-(\phi o)^2}{2(k-o)}}do\bigg] \\
 \sim\; \phi k^{-1.25+\frac{1}{2\delta}}\; \sim\; k^{-1.25-\frac{1}{2\delta}}.
 \end{aligned}
\end{align}
This rough calculation shows that $\delta= 2$ is where the distribution departs from a $1.5$ power law.

\end{subsection}

\begin{subsection}{Large and small input regimes}\label{Large and Small input regimes}
Suppose we have a $(N,p)$-BB space. 

 With $\lambda \le N$ an integer parameter we will define a new random variable $X_{N,p,\lambda}$. For a configuration $\omega$, a second configuration $\omega'$ is defined by the following procedure.
 When $j \le \lambda$,
\begin{align*}
a_{M,j}\left(\omega'\right)\;=1\quad \& \quad  a_{ij}\left(\omega'\right)\;=0, \; \forall \; i <M.
\end{align*}
When $j > \lambda $, set $a_{i,j}\left(\omega'\right)\;=\; a_{i,j}(\omega)$, for all $i$.\newline

 Define $X_{N,p,\lambda}(\omega)= A_{N,p}(\omega')- \lambda $. We can derive the following using Theorem \eqref{Basic_enumeration_theorem_moderate_input} :
\begin{thm}\label{Basic_enumeration_theorem_large_input}
$P(X_{N,p,\lambda}=k)= {N-\lambda \choose k} p^{k}(\lambda+1)(k+\lambda+1)^{k-1}
(1-(k+\lambda +1)p)^{N-\lambda-k}$
\end{thm}
\textbf{\underline{Small input regime}}
Here we choose $\lambda = \lambda_{0}$, where $\lambda_0$ is some constant independent of $N$. Using Theorem \eqref{Basic_enumeration_theorem_large_input}, and Stirling's formula  (like with the avalanche distribution) we show that as $N$ and $k$ grow to infinity with $\frac{k}{N} \to 0$,  $P(X_{N,p,\lambda}=k) = \Theta (k^{-1.5})$.

\textbf{\underline{Large input regime}}
Let us now put $\lambda=\hat{\lambda}_N= \lambda_{0}\times N$, $\lambda_{0}< 1$. We will also demand that $\alpha \times (1+\lambda_0) < 1$. This implies there is massive external input during firing that forces a proportion of the system to spontaneously fire.
Now observe that $X_{N,p,\hat{\lambda}_{N}}$ has the distribution of Quasi Binomial 1 distribution (see \cite{consul1974simple}) with $n=(1-\lambda_0)N,\; a=(1+\hat{\lambda}_N)p, \; \theta=p,\; b=1-n\theta-a$. Again it is known as $n \rightarrow +\infty$, a Quasi Binomial 1 distribution approaches the Generalized Poisson distribution (\cite{consul1973generalization}), which is a type of Lagrangian distribution (\cite{consul1972use}, \cite{consul2006lagrangian}). It has further been established that Lagrangian random variables approach the standard normal distribution under certain conditions (see \cite{consul1973some}). All of these together lead to the following theorem.
\begin{thm}
$\frac{X_{N,\frac{\alpha}{N},\hat{\lambda}_{N}}-\mu_{N}}{\sigma_{N}}$ converges in distribution to a standard normal variable as $N$ goes to $\infty$, where $$\mu_{N}=\frac{\alpha (1+\lambda_0 N)(1-\lambda_0)}{1- \alpha(1-\lambda_0)}, \quad (\sigma_{N})^{2}=\frac{\alpha (1+\lambda_0 N)(1-\lambda_0)}{(1- \alpha(1-\lambda_0))^{3}}.$$
\end{thm}
\end{subsection}
\end{section}
\bibliographystyle{plain}
\bibliography{soc_bib}
\end{document}